\documentclass[11pt,letterpaper]{article}
\pdfoutput=1
\usepackage[hmargin=1.3in,vmargin={1.1in, 1.2in}]{geometry}
\usepackage{microtype,amsmath,amsthm}
\usepackage{enumitem,amssymb}
\usepackage{showlabels}
\usepackage{algorithm}
\usepackage{algpseudocode}
\algrenewcommand\algorithmicensure{\textbf{Output:}}
\algrenewcommand\algorithmicrequire{\textbf{Input:}}

\usepackage{tikz}
\usetikzlibrary{shapes.geometric,arrows,positioning}

\usepackage{tcs-macros}
\numberwithin{equation}{section}

\DeclareMathOperator{\KL}{KL}
\newcommand{\KLM}{\KL_{\rm min}}
\newcommand{\kl}[2]{\KL\left(#1\,\middle\|\,#2\right)}
\newcommand{\klm}[3]{\KL_{\rm min}^{#1}\left(#2\,\middle\|\,#3\right)}
\newcommand{\skl}[3]{\KL^*_{#1}\left(#2\,\middle\|\,#3\right)}

\DeclareMathOperator{\Ent}{H}
\newcommand{\ent}[1]{\Ent\left(#1\right)}
\newcommand{\sent}[2]{\Ent^*_{#1}\left(#2\right)}

\DeclareMathOperator{\Sim}{\mathsf{Sim}}

\def\mydefbb#1{\expandafter\def\csname t#1\endcsname{\widetilde{#1}}}
\def\mydefallbb#1{\ifx#1\mydefallbb\else\mydefbb#1\expandafter\mydefallbb\fi}
\mydefallbb RrYyXxA\mydefallbb
\def\mydefbb#1{\expandafter\def\csname h#1\endcsname{\widehat{#1}}}
\def\mydefallbb#1{\ifx#1\mydefallbb\else\mydefbb#1\expandafter\mydefallbb\fi}
\mydefallbb GRrYy\mydefallbb
\def\mydefbb#1{\expandafter\def\csname bo#1\endcsname{\overline{#1}}}
\def\mydefallbb#1{\ifx#1\mydefallbb\else\mydefbb#1\expandafter\mydefallbb\fi}
\mydefallbb YR\mydefallbb

\newcommand{\tG}{\mathsf{\widetilde{G}}}

\newcommand{\w}{\mathsf{w}}
\newcommand{\klgap}{\Delta}
\newcommand{\klfail}{\delta}
\newcommand{\invp}{\eps}

\newcommand{\PPT}{\textsc{ppt}}

\newcommand\email[1]{\href{mailto:#1}{{\ttfamily #1}}}

\title{Unifying computational entropies\\via Kullback--Leibler divergence}
\author{
	Rohit Agrawal
	\thanks{Harvard John A. Paulson School of Engineering and Applied
		Sciences. Supported by the Department of Defense (DoD) through the
		National Defense Science \& Engineering Graduate Fellowship (NDSEG)
		Program.}\\
	\email{rohitagr@seas.harvard.edu}
	\and
	Yi-Hsiu Chen
	\thanks{Harvard John A. Paulson School of Engineering and Applied
		Sciences. Supported by NSF grant CCF-1763299.}\\
	\email{yhchen@seas.harvard.edu}
	\and
	Thibaut Horel
	\thanks{Harvard John A. Paulson School of Engineering and Applied Sciences.
		Supported in  part  by  the  National  Science Foundation  under grants
		CAREER  IIS-1149662, CNS-1237235 and CCF-1763299, by the Office of
		Naval Research under grants YIP N00014-14-1-0485 and N00014-17-1-2131,
		and by a Google Research Award.}\\
	\email{thorel@seas.harvard.edu}
	\and
	Salil Vadhan
	\thanks{Harvard John A. Paulson School of Engineering and Applied
		Sciences. Supported by NSF grant CCF-1763299.}\\
	\email{salil\_vadhan@harvard.edu}
}

\begin{document}

\maketitle

\begin{abstract}
We introduce {\em hardness in relative entropy}, a new notion of hardness for
search problems which on the one hand is satisfied by all one-way functions and
on the other hand implies both \emph{next-block pseudoentropy} and
\emph{inaccessible entropy}, two forms of computational entropy used in recent
constructions of pseudorandom generators and statistically hiding commitment
schemes, respectively.  Thus, hardness in relative entropy unifies the latter
two notions of computational entropy and sheds light on the apparent
``duality'' between them.  Additionally, it yields a more modular and
illuminating proof that one-way functions imply next-block inaccessible
entropy, similar in structure to the proof that one-way functions imply
next-block pseudoentropy (Vadhan and Zheng, STOC `12).
\end{abstract}

\vfill
\textbf{Keywords:} one-way function, pseudorandom generator, pseudoentropy, computational entropy, inaccessible entropy, statistically hiding commitment, next-bit pseudoentropy.

\newpage
\section{Introduction}

\subsection{One-way functions and computational entropy}

One-way functions~\cite{DiffieHe76} are on one hand the minimal assumption for
complexity-based cryptography~\cite{ImpagliazzoLu89}, but on the other hand can
be used to construct a remarkable array of cryptographic primitives, including
such powerful objects as CCA-secure symmetric encryption, zero-knowledge proofs
and statistical zero-knowledge arguments for all of $\NP$, and secure
multiparty computation with an honest
majority~\cite{GoldreichGoMi86,GoldreichMiWi91,GoldreichMiWi87,HastadImLeLu99,Rompel90,Naor91,HaitnerNgOnReVa09}.
All of these constructions begin by converting the ``raw hardness'' of
a one-way function (OWF) to one of the following more structured cryptographic
primitives: a pseudorandom generator (PRG)~\cite{BlumM82,Yao82B}, a universal
one-way hash function (UOWHF)~\cite{NaorYu89}, or a statistically hiding
commitment scheme (SHC)~\cite{BrassardChCr88}.

The original constructions of these three primitives from
arbitrary one-way functions~\cite{HastadImLeLu99,Rompel90,HaitnerNgOnReVa09}
were all very complicated and inefficient.  Over the past decade, there has
been a series of simplifications and efficiency improvements to these
constructions~\cite{HaitnerReVaWe09,HRV13,HHRVW10,VZ12},
leading to a situation where the constructions of two of these primitives ---
PRGs and SHCs --- share a very similar structure and seem ``dual'' to each other.
Specifically, these constructions proceed as follows:
\begin{enumerate}
	\item Show that every OWF $f : \zo^n\rightarrow \zo^n$ has a gap between
		its ``real entropy'' and an appropriate form of ``computational
		entropy''.  Specifically, for constructing PRGs, it is shown that the
		function $\sG(x)=(f(x),x_1,x_2,\ldots,x_n)$ has ``next-block
		pseudoentropy'' at least $n+\omega(\log n)$ while its real entropy is
		$\ent{\sG(U_n)} = n$~\cite{VZ12} where $\ent{\cdot}$ denotes Shannon
		entropy. For constructing SHCs, it is shown that the function $\sG(x)
		= (f(x)_1,\ldots,f(x)_n,x)$ has ``next-block accessible entropy'' at
		most $n-\omega(\log n)$ while its real entropy is again
		$\ent{\sG(U_n)}=n$~\cite{HaitnerReVaWe09}.  Note that the differences
		between the two cases are whether we break $x$ or $f(x)$ into
		individual bits (which matters because the ``next-block'' notions of
		computational entropy depend on the block structure) and  whether the
		form of computational entropy is larger or smaller than the real
		entropy.

	\item An ``entropy equalization'' step that converts $\sG$ into a similar
		generator where the real entropy in each block conditioned on the
		prefix before it is known.  This step is exactly the same in both
		constructions.

	\item A ``flattening'' step that converts the (real and computational)
		Shannon entropy guarantees of the generator into ones on (smoothed)
		min-entropy and max-entropy.  This step is again exactly the same in
		both constructions.

	\item A ``hashing'' step where high (real or computational) min-entropy is
		converted to uniform (pseudo)randomness and low (real or computational)
		max-entropy is converted to a small-support or disjointness property.
		For PRGs, this step only requires randomness
		extractors~\cite{HastadImLeLu99,NisanZu96}, while for SHCs it requires
		(information-theoretic) interactive
		hashing~\cite{NaorOsVeYu98,DingHRS04}.  (Constructing full-fledged SHCs
		in this step also utilizes UOWHFs, which can be constructed from
		one-way functions~\cite{Rompel90}.  Without UOWHFs, we obtain a weaker
		binding property, which nevertheless suffices for constructing
		statistical zero-knowledge arguments for all of $\NP$.)
\end{enumerate}
This common construction template came about through a back-and-forth exchange of ideas between the two lines of work.  Indeed, the
uses of computational entropy notions, flattening, and hashing originate with PRGs~\cite{HastadImLeLu99}, whereas the ideas of using next-block notions, obtaining them from breaking $(f(x),x)$ into short blocks, and entropy equalization originate with SHCs~\cite{HaitnerReVaWe09}.  All this leads to a feeling that the two constructions, and their underlying computational entropy notions, are ``dual'' to each other and should be connected at a formal level.

In this paper, we make progress on this project of unifying the notions of
computational entropy, by introducing a new computational entropy notion that
yields both next-block pseudoentropy and next-block accessible entropy in
a clean and modular fashion. It is inspired by the proof of \cite{VZ12} that
$(f(x),x_1,\ldots,x_n)$ has next-block pseudoentropy $n+\omega(\log n)$, which
we will describe now.

\subsection{Next-block pseudoentropy via relative pseudoentropy}
\label{sec:nbpe}

We recall the definition of next-block pseudoentropy, and the result of
\cite{VZ12} relating it to one-wayness.

\begin{definition}[next-block pseudoentropy, informal]\label{defn:nbpe-intro}
Let $n$ be a security parameter, and $X=(X_1,\ldots,X_m)$ be a random variable distributed on strings of length $\poly(n)$.
We say that $X$ has {\em next-block pseudoentropy} at least $k$ if there is a random variable $Z=(Z_1,\ldots,Z_m)$, jointly
distributed with $X$, such that:
\begin{enumerate}
\item For all $i=1,\ldots,m$, $(X_1,\ldots,X_{i-1},X_i)$ is computationally indistinguishable from $(X_1,\ldots,X_{i-1},Z_i)$.
\item $\sum_{i=1}^m \ent{Z_i|X_1,\ldots,X_{i-1}} \geq k$.
\end{enumerate}
Equivalently, for $I$ uniformly distributed in $[m]$, $X_I$ has \emph{conditional pseudoentropy} at least $k/m$ given
$(X_1,\ldots,X_{i-1})$.
\end{definition}

It was conjectured in \cite{HaitnerReVa10} that next-block pseudoentropy could be obtained from any OWF by breaking its input into bits, and this conjecture was proven in \cite{VZ12}:

\begin{theorem}[\cite{VZ12}, informal] \label{thm:VZ-intro}
Let $f : \zo^n\rightarrow \zo^n$ be a one-way function, let $X$ be uniformly
distributed in $\zo^n$, and let $X=(X_1,\ldots,X_m)$ be a partition of $X$ into
blocks of length $O(\log n)$.  Then $(f(X),X_1,\ldots,X_m)$ has next-block
pseudoentropy at least $n+\omega(\log n)$.
\end{theorem}

The intuition behind Theorem~\ref{thm:VZ-intro} is that since $X$ is hard to sample given $f(X)$, then it
should have some extra computational entropy given $f(X)$.  This intuition is formalized using the following notion of ``relative pseudoentropy,'' which is a renaming of \cite{VZ12}'s notion of ``KL-hard for sampling,''
to better unify the terminology with the notions introduced in this work.

\begin{definition}[relative pseudoentropy] \label{def:KL-hard-to-sample}
Let $n$ be a security parameter, and $(X,Y)$ be a pair of random variables, jointly distributed over strings of length
$\poly(n)$.  We say that $X$ has {\em relative pseudoentropy at least $\klgap$} given $Y$ if for all probabilistic polynomial-time
$\sS$, we have
$$\kl{X,Y}{\sS(Y),Y} \geq \klgap,$$
where $\kl{\cdot}{\cdot}$ denotes the relative entropy (a.k.a.\ Kullback--Leibler divergence).\footnote{Recall that for random variables $A$ and $B$ with
	$\Supp(A)\subseteq\Supp(B)$, the
	relative entropy is defined by $\kl{A}{B}=\ex[a\from
A]{\log(\pr{A=a}/\pr{B=a})}$.}
\end{definition}
That is, it is hard for any efficient adversary $\sS$ to sample the conditional distribution of $X$ given $Y$, even approximately.

The first step of the proof of Theorem~\ref{thm:VZ-intro} is to show that
one-wayness implies relative pseudoentropy (which can be done with a one-line
calculation):
\begin{lemma} \label{lem:OWF-KL-hard-to-sample-intro}
	Let $f : \zo^n\rightarrow \zo^n$ be a one-way function and let $X$ be
	uniformly distributed in $\zo^n$.   Then $X$ has relative pseudoentropy at
	least $\omega(\log n)$ given $f(X)$.
\end{lemma}

Next, we break $X$ into short blocks, and show that the relative pseudoentropy is preserved:
\begin{lemma}\label{lem:OWF-nb-KL-hard-to-sample-intro}
Let $n$ be a security parameter, let $(X,Y)$ be random variables
distributed on strings of length $\poly(n)$, let $X=(X_1,\ldots,X_m)$ be
a partition of $X$ into blocks, and let $I$ be uniformly distributed in $[m]$.
If $X$ has relative pseudoentropy at least $\klgap$ given $Y$, then $X_I$ has relative pseudoentropy at least $\klgap/m$ given $(Y,X_1,\ldots,X_{I-1})$.
\end{lemma}

Finally, the main part of the proof is to show that, once we have short blocks, relative pseudoentropy is {\em equivalent} to a gap between
conditional pseudoentropy and real conditional entropy.
\begin{lemma} \label{lem:VZ-characterizing-intro}
Let $n$ be a security parameter, $Y$ be a random variable distributed on
strings of length $\poly(n)$, and $X$ a random variable distributed on strings
of length $O(\log n)$.   Then $X$ has relative pseudoentropy  at least $\klgap$ given $Y$ iff
$X$ has conditional pseudoentropy at least $\ent{X|Y}+\klgap$ given $Y$.
\end{lemma}

Putting these three lemmas together, we see that when $f$ is a one-way function, and we break $X$ into blocks of length $O(\log n)$ to obtain $(f(X),X_1,\ldots,X_m)$, on average, the conditional pseudoentropy of $X_I$
given $(f(X),X_1,\ldots,X_{I-1})$ is larger than its real conditional entropy by
$\omega(\log n)/m$.  This tells us that the next-block pseudoentropy
of $(f(X),X_1,\ldots,X_m)$ is larger than its real entropy by $\omega(\log n)$, as claimed in Theorem~\ref{thm:VZ-intro}.

We remark that Lemma~\ref{lem:VZ-characterizing-intro} explains why we need
to break the input of the one-way function into short blocks:  it is false when $X$ is long.  Indeed, if $f$ is a one-way function, then we have
already seen that $X$ has $\omega(\log n)$ relative pseudoentropy given $f(X)$
(Lemma~\ref{lem:OWF-KL-hard-to-sample-intro}), but it
does not have conditional pseudoentropy noticeably larger than $\ent{X|f(X)}$ given $f(X)$ (as correct preimages can be efficiently distinguished from incorrect ones using $f$).

\subsection{Inaccessible entropy}

As mentioned above, for constructing SHCs from one-way functions, the
notion of next-block pseudoentropy is replaced with next-block accessible entropy:

\begin{definition}[next-block accessible entropy, informal]
Let $n$ be a security parameter, and $Y=(Y_1,\ldots,Y_m)$ be a random variable distributed on strings of length $\poly(n)$.
We say that $Y$ has {\em next-block accessible entropy} at most $k$ if
the following holds.

Let $\tG$ be any probabilistic $\poly(n)$-time algorithm that takes a sequence of uniformly random
strings $\tR=(\tR_1,\ldots,\tR_m)$ and outputs
a sequence $\tY=(\tY_1,\ldots,\tY_m)$ in an ``online fashion''
by which we mean that $\tY_i = \tG(\tR_1,\ldots,\tR_i)$ depends
on only the first $i$ random strings of $\tG$ for $i=1,\ldots,m$.
Suppose further that $\Supp(\tY)\subseteq \Supp(Y)$.

Then we require:
$$\sum_{i=1}^m \ent{\tY_i | \tR_1,\ldots,\tR_{i-1}} \leq k.$$
\end{definition}

(Next-block) accessible entropy differs from (next-block) pseudoentropy in two ways:
\begin{enumerate}
	\item Accessible entropy is useful as an {\em upper} bound on computational
		entropy, and is interesting when it is {\em smaller} than the real
		entropy $\ent{Y}$.  We refer to the gap $\ent{Y}-k$ as the
	{\em next-block inaccessible entropy} of $Y$.
	\item The accessible entropy adversary $\tG$ is trying to {\em
		generate} the random variables $Y_i$ conditioned on the history
		rather than recognize them.  Note that we take the ``history''
		to not only be the previous blocks $(\tY_1,\ldots,\tY_{i-1})$,
		but the coin tosses $(\tR_1,\ldots,\tR_{i-1})$ used to generate
		those blocks.
\end{enumerate}
Note that one unsatisfactory aspect of the definition is that
		when the random variable $Y$ is not {\em flat} (i.e. uniform on its
		support), then there can be an adversary $\tG$ achieving
		accessible entropy even {\em larger} than $\ent{Y}$, for example
		by making $\tY$ uniform on $\Supp(Y)$.

Similarly to (and predating) Theorem~\ref{thm:VZ-intro}, it is known
that one-wayness implies next-block inaccessible entropy.

\begin{theorem}[\cite{HaitnerReVaWe09}] \label{thm:HRVW-intro}
Let $f : \zo^n\rightarrow \zo^n$ be a one-way function, let $X$ be uniformly distributed in $\zo^n$, and let $(Y_1,\ldots,Y_m)$ be
a partition of $Y=f(X)$ into blocks of length $O(\log n)$.  Then
$(Y_1,\ldots,Y_m,X)$ has next-block accessible entropy at most $n-\omega(\log n)$.
\end{theorem}

Unfortunately, however, the existing proof of
Theorem~\ref{thm:HRVW-intro} is not modular like that of
Theorem~\ref{thm:VZ-intro}.  In particular, it does not isolate the step
of relating one-wayness to entropy-theoretic measures (like
Lemma~\ref{lem:OWF-KL-hard-to-sample-intro} does) or the significance of
having short blocks (like Lemma~\ref{lem:VZ-characterizing-intro} does).

\subsection{Our results}

We remedy the above state of affairs by providing a new, more general
notion of hardness in relative entropy that allows us to obtain next-block inaccessible entropy in a modular way while also encompassing what is needed for next-block pseudoentropy.

Like in relative pseudoentropy, we will consider a pair of jointly distributed random variables $(Y,X)$.  Following the spirit of accessible entropy, the adversary $\tG$ for our new notion will try to {\em generate} $Y$ together with $X$, rather than taking $Y$ as input.  That is, $\tG$ will
take randomness $\tR$ and output a pair
$(\tY,\tX) = \tG(\tR) = (\tG_1(\tR),\tG_2(\tR))$, which we require
to be always within the support of $(Y,X)$.  Note that $\tG$ need not be
an online generator; it can generate both $\tY$ and $\tX$ using the same randomness $\tR$.   Of course, if $(Y,X)$ is efficiently samplable (as it would be in most cryptographic applications), $\tG$ could generate
$(\tY,\tX)$ identically distributed to $(Y,X)$ by just using the ``honest''
sampler $\sG$ for $(Y,X)$.   So, in addition, we require that the adversary $\tG$
also come with a {\em simulator} $\sS$, that can simulate its coin tosses given
only $\tY$.
The goal of the adversary is to minimize
the relative entropy
$$\kl{\tR,\tY}{\sS(Y),Y}$$
for a uniformly random $\tR$. This divergence measures both how well $\tG_1$
approximates the distribution of $Y$ as well as how well $\sS$ simulates the
corresponding coin tosses of $\tG_1$.  Note that when $\tG$ is the honest
sampler $\sG$, the task of $\sS$ is exactly to sample from the conditional
distribution of $\tR$ given $\sG_1(\tR) = Y$.  However, the adversary may reduce
the divergence by instead designing the sampler $\tG$ and simulator $\sS$ to
work in concert, potentially trading off how well $\tG(\tR)$ approximates $Y$
in exchange for easier simulation by $\sS$.  Explicitly, the definition is as
follows.
\begin{definition}[hardness in relative entropy, informal version of
	Definition~\ref{def:kl-hard}]\label{def:KL-hard-intro}
Let $n$ be a security parameter, and $(Y,X)$ be a pair of random variables jointly distributed over strings of length
$\poly(n)$.  We say that $(Y,X)$ has \emph{hardness at least $\klgap$ in relative entropy} if the following holds.

Let $\tG=(\tG_1,\tG_2)$ and $\sS$ be probabilistic $\poly(n)$-time algorithms
such that $\Supp(\tG(\tR))\subseteq \Supp((Y,X))$, where $\tR$ is uniformly
distributed. Then writing $\tY = \tG_1(\tR)$, we require that
$$\kl{\tR, \tY}{\sS(Y), Y} \geq \klgap.$$
\end{definition}

Similarly to Lemma~\ref{lem:OWF-KL-hard-to-sample-intro}, we can
show that one-way functions achieve this notion of hardness in relative entropy.
\begin{lemma} \label{lem:OWF-KL-hard-intro}
Let $f : \zo^n\rightarrow \zo^n$ be a one-way function and let $X$ be uniformly distributed in $\zo^n$.   Then
$(f(X),X)$ has hardness $\omega(\log n)$ in relative entropy.
\end{lemma}
Note that this lemma implies Lemma~\ref{lem:OWF-KL-hard-to-sample-intro}.
If we take $\tG$ to be the ``honest'' sampler $\tG(x)=(f(x),x)$, then we
have:
$$\kl{X,f(X)}{\sS(Y),Y} = \kl{\tR,\tY}{\sS(Y),Y},$$
which is is $\omega(\log n)$ by Lemma~\ref{lem:OWF-KL-hard-intro}. That is,
relative pseudoentropy (as in
Definition~\ref{def:KL-hard-to-sample} and
Lemma~\ref{lem:OWF-KL-hard-to-sample-intro}) is obtained
by fixing $\tG$ and focusing on the hardness for the simulator $\sS$, i.e. the
divergence $\kl{X, Y}{S(Y), Y}$. Furthermore, the step of breaking into short
blocks (Lemma~\ref{lem:OWF-nb-KL-hard-to-sample-intro}) is equivalent to
requiring the simulator be \emph{online} and
showing that relative pseudoentropy implies the following notion of
\emph{next-block relative pseudoentropy}:
\begin{definition}[next-block relative pseudoentropy, informal]
	Let $n$ be a security parameter, $(X,Y)$ be jointly distributed random
	variables over strings of length $\poly(n)$, and let $X = (X_1, \dots,
	X_m)$ be a partition of $X$ into blocks. We say that $X$ has {\em
	next-block relative pseudoentropy at least $\klgap$ given $Y$} if for
	all probabilistic polynomial-time $\sS$, we have
	\[\sum_{i=1}^m \kl{X_i|X_{< i},Y}{\sS(X_{<i},Y)|X_{<i},Y} \geq \klgap,\]
	where we use the notation $z_{<i}\eqdef(z_1,\ldots,z_{i-1})$.

	Here, the simulator $\sS$ is required to be ``online'' in the sense
	that it cannot simulate $(X_1, \dots, X_m)$ at once, but must simulate
	$X_i$ only as a function of $X_{<i}$ and $Y$.
\end{definition}

In particular, Lemma~\ref{lem:VZ-characterizing-intro} is thus equivalent to
the statement that having next-block relative pseudoentropy at least $\klgap$
for blocks of length $O(\log n)$ is equivalent to having next-block
pseudoentropy at least $\klgap + \sum_{i=1}^m \ent{X_i|X_{<i},Y}$ in the sense
of Definition~\ref{defn:nbpe-intro}.

Conversely, we show that inaccessible entropy arises from hardness in relative
entropy by first requiring the \emph{generator} $\sG$ to be online and breaking
the relative entropy into blocks to obtain the following next-block hardness
property.

\begin{definition}[next-block hardness in relative entropy, informal]
Let $n$ be a security parameter, and $Y=(Y_1,\ldots,Y_m)$ be a random
variable distributed on strings of length $\poly(n)$.  We say that $Y$
has {\em next-block hardness at least $\klgap$ in relative entropy} if the following
holds.

Let $\tG$ be any probabilistic $\poly(n)$-time algorithm that takes a sequence
of uniformly random strings $\tR=(\tR_1,\ldots,\tR_m)$ and outputs a sequence
$\tY=(\tY_1,\ldots,\tY_m)$ in an ``online fashion'' by which we mean that
$\tY_i = \tG(\tR_1,\ldots,\tR_i)$ depends on only the first $i$ random strings
of $\tG$ for $i=1,\ldots,m$.  Suppose further that $\Supp(\tY)\subseteq
\Supp(Y)$.  Additionally, let $\sS$ be a probabilistic $\poly(n)$-time
algorithms such for all $i=1,\dots,m$, $\sS$ takes as input
$\hR_1,\dots,\hR_{i-1}$ and $Y_i$ and outputs $\hR_i$, where $\hR_j$ has the
same length as $\tR_j$. Then we require that for all such
$(\tG,\sS)$, we have:
$$\sum_{i=1}^m \kl{\tR_i, \tY_i | \tR_{<i},\tY_{<i}}{ \hR_i, Y_i | \hR_{<i}, Y_{<i}} \geq \klgap.$$
\end{definition}

Observe that hardness in relative entropy can be seen as the specific case of
next-block hardness in relative entropy when there is only one block
(\emph{i.e.,} setting $m=1$ in the previous definition).

Next, we fix the \emph{simulator}, analogously to how relative pseudoentropy
was obtained by fixing the generator, and obtain \emph{next-block inaccessible
relative entropy}:
\begin{definition}[next-block inaccessible relative entropy,
	informal] Let $n$ be a security parameter, and $Y=(Y_1,\ldots,Y_m)$ be
	a random variable distributed on strings of length $\poly(n)$.  We say that
	$Y$ has {\em next-block inaccessible relative entropy at least $\klgap$} if the
	following holds.

Let $\tG$ be any probabilistic $\poly(n)$-time algorithm that takes a sequence of uniformly random
strings $\tR=(\tR_1,\ldots,\tR_m)$ and outputs
a sequence $\tY=(\tY_1,\ldots,\tY_m)$ in an online fashion,
and such that $\Supp(\tY)\subseteq \Supp(Y)$.
Then we require that for all such $\tG$, we have:
$$\sum_{i=1}^m \kl{\tY_i | \tR_{<i},\tY_{<i}}{ Y_i |
R_{<i}, Y_{<i}} \geq \klgap,$$
where $R=(R_1,\ldots,R_m)$ is a dummy random variable independent of $Y$.
\end{definition}
That is, the goal of the online generator $\tG$ is to generate $\tY_i$
given the history of coin tosses $\tR_{<i}$ with the same conditional
distribution as $Y_i$ given $Y_{<i}$. As promised, there is no explicit
simulator in the definition of next-block inaccessible relative entropy,
as we essentially dropped all $\hR$ variables from the definition of
next-block hardness in relative entropy.
Nevertheless we
can obtain it from hardness in relative entropy by using sufficiently short blocks:

\begin{lemma} \label{lem:rejection-sampling-intro}
Let $n$ be a security parameter, let $Y$ be a random variable
distributed on strings of length $\poly(n)$, and let
$Y=(Y_1,\ldots,Y_m)$ be a partition of $Y$ into blocks of length $O(\log
n)$.

If $(Y_1,\dots,Y_m)$ has next-block hardness at least $\klgap$ in relative entropy,
then $(Y_1,\ldots,Y_m)$ has next-block inaccessible relative entropy at
least $\klgap-\negl(n)$.
\end{lemma}
An intuition for the proof is that since the blocks are of logarithmic length, given $Y_i$ we can
simulate the corresponding coin tosses of $\tR_i$ of $\tG$ by rejection sampling and succeed with high probability in $\poly(n)$ tries.

A nice feature of the definition of next-block inaccessible relative entropy
compared to inaccessible entropy is that it is meaningful even for non-flat
random variables, as the Kullback--Leibler divergence is always nonnegative.  Moreover, for flat
random variables, it equals the inaccessible entropy:

\begin{lemma} \label{lem:KL-to-inaccessible-intro}
Suppose $Y=(Y_1,\ldots,Y_m)$ is a flat random variable.  Then $Y$ has
next-block inaccessible relative entropy at least $\klgap$ if and only
if $Y$ has accessible entropy at most $\ent{Y}-\klgap.$
\end{lemma}
Intuitively, this lemma comes from the identity that if $Y$ is a flat
random variable and $\Supp(\tY)\subseteq \Supp(Y)$, then $\ent{\tY} =
\ent{Y}-\kl{\tY}{Y}$. We stress that we do not require the individual
blocks $Y_i$ have flat distributions, only that the random variable $Y$
as a whole is flat. For example, if $f$ is a function and $X$ is
uniform, then $(f(X), X)$ is flat even though $f(X)$ itself may be far
from flat.

Putting together Lemmas~\ref{lem:OWF-KL-hard-intro},
\ref{lem:rejection-sampling-intro}, and \ref{lem:KL-to-inaccessible-intro}, we
obtain a new, more modular (and slightly tighter) proof of Theorem~\ref{thm:HRVW-intro}.  The
reduction implicit in the combination of these lemmas is the same as the one in
\cite{HaitnerReVaWe09}, but the analysis is different. (In particular,
\cite{HaitnerReVaWe09} makes no use of KL divergence.) Like the existing proof
of Theorem~\ref{thm:VZ-intro}, this proof separates the move from one-wayness
to a form of hardness involving relative entropies, the role of short blocks,
and the move from hardness in relative entropy to computational entropy,
as summarized in Figure~\ref{fig:notions-diagram}.
Moreover, this further illumination of and toolkit for notions of computational
entropy may open the door to other applications in cryptography.
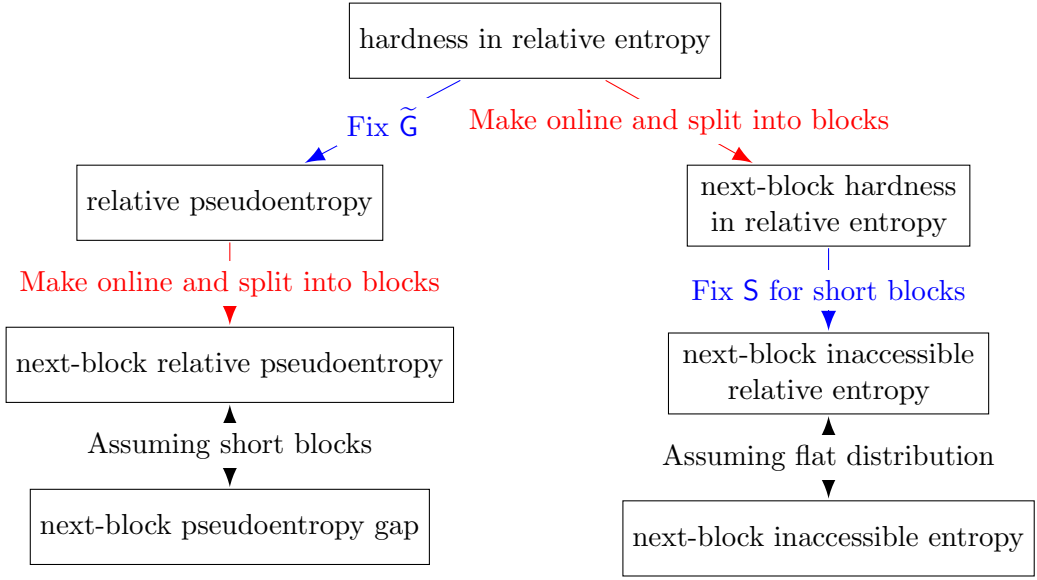
\begin{figure}[htbp]
\centering
\begin{tikzpicture}
	\tikzset{>=latex,ultra thick}
	\tikzstyle{pred} = [rectangle, minimum width=2.5cm, minimum height=1cm, text centered, draw=black];
	\tikzstyle{fix} = [draw,->,blue];
	\tikzstyle{online} = [draw,->,red]
	\tikzstyle{equiv} = [draw, <->];
	\node (hire) [pred] {hardness in relative entropy};
	\node (relpe) [pred,below left=1.1cm and -0.5cm of hire] {relative pseudoentropy};
	\node (nbhre) [pred,text width=3.5cm, below right=1.1cm and -0.5cm of hire] {next-block hardness in relative entropy};
	\node (nbrpe) [pred,below=1.1cm of relpe] {next-block relative pseudoentropy};
	\node (nbire) [pred,text width=4cm,below=1.1cm of nbhre] {next-block inaccessible\\ relative entropy};
	\node (nbpe) [pred, below=1.1cm of nbrpe] {next-block pseudoentropy gap};
	\node (nbie) [pred, below=1.1cm of nbire] {next-block inaccessible entropy};

	\path [fix] (hire) edge  node[fill=white] {Fix $\tG$} (relpe);
	\path [online] (hire) edge  node[fill=white] {Make online and split into blocks} (nbhre);
	\path [online] (relpe) edge node[fill=white] {Make online and split into blocks} (nbrpe);
	\path [fix] (nbhre) edge node[fill=white] {Fix $\sS$ for short blocks} (nbire);
	\path [equiv] (nbrpe) edge node[fill=white] {Assuming short blocks} (nbpe);
	\path [equiv] (nbire) edge node[fill=white] {Assuming flat distribution} (nbie);
\end{tikzpicture}
\caption{Relationships between hardness notions.}
\label{fig:notions-diagram}
\end{figure}

We remark that another interesting direction for future work is to find
a construction of universal one-way hash functions (UOWHFs) from one-way
functions that follows a similar template to the above constructions of PRGs
and SHCs.  There is now a construction of UOWHFs based on a variant of
inaccessible entropy~\cite{HHRVW10}, but it remains more complex and
inefficient than those of PRGs and SHCs.

\section{Preliminaries}

\paragraph{Notations.} For a tuple $x = (x_1,\dots,x_n)$, we write $x_{\leq i}$ for $(x_1,\dots,x_i)$,
and $x_{<i}$ for $(x_1,\dots,x_{i-1})$.

$\poly$ denotes the set of polynomial functions and $\negl$ the set of all
negligible functions: $\eps\in\negl$ if for all $p\in\poly$ and large enough
$n\in\bbN$, $\eps(n)\leq 1/p(n)$. We will sometimes abuse notations and write
$\poly(n)$ to mean $p(n)$ for some $p\in\poly$ and similarly for $\negl(n)$.

$\PPT{}$ stands for probabilistic polynomial time and can be either in the
uniform or non-uniform model of computation. All our results are stated as
uniform polynomial time oracle reductions and are thus meaningful in both
models.

For a random variable $X$ over $\cX$, $\Supp(X)\eqdef \{x\in\cX\,:\,
\Pr[X=x]>0\}$ denotes the support of $X$. A random variable is \emph{flat} if
it is uniform over its support. Random variables will be written with uppercase
letters and the associated lowercase letter represents a generic element from
its support.

\paragraph{Information theory.}

\begin{definition}[Entropy]
	For a random variable $X$ and $x\in\Supp(X)$, the \emph{sample entropy}
	(also called surprise) of $x$ is $\sent{x}{X}\eqdef \log(1/\pr{X=x})$.  The
	\emph{entropy} $\ent{X}$ of $X$ is the expected sample entropy:
	$\ent{X} \eqdef \ex[x\from X]{\sent{x}{X}}$.
\end{definition}

\begin{definition}[Conditional entropy]
	Let $(A, X)$ be a pair of random variables and consider $(a,x)\in\Supp(A,
	X)$, the \emph{conditional sample entropy} of $(a, x)$ is
	$\sent{a,x}{A|X}\eqdef \log(1/\pr{A=a\,|\,X=x})$ and the \emph{conditional
	entropy} of $A$ given $X$ is the expected conditional sample entropy:
	\begin{displaymath}
	\ent{A| X}\eqdef \ex[(a,x)\from (A,X)]{\log\frac{1}{\pr{A=a\,|\,X=x}}}\;.
	\end{displaymath}
\end{definition}

\begin{proposition}[Chain rule for entropy] Let $(A, X)$ be a pair of random
	variables, then $\ent{A, X} = \ent{A|X} + \ent{X}$ and for
	$(a,x)\in\Supp(A, X)$, $\sent{a,x}{A,X} = \sent{a,x}{A|X} + \sent{x}{X}$.
\end{proposition}

\begin{definition}[Relative entropy\footnote{\emph{Relative entropy} is also
	commonly referred to as \emph{Kullback--Liebler divergence}, which explains
the standard $\KL$ notation. We prefer to use relative entropy to have more
uniformity across the notions discussed in this work.}]
	\label{def:kl}
	For a pair $(A, B)$ of random variables and $(a, b)\in\Supp(A, B)$ the
	\emph{sample relative entropy} ($\log$-probability ratio) is:
	\begin{displaymath}
		\skl{a}{A}{B}\eqdef \log\frac{\pr{A=a}}{\pr{B=a}}\;,
	\end{displaymath}
	and the \emph{relative entropy} of $A$ with respect to $B$ is the expected
	sample relative entropy:
	\begin{displaymath}
		\kl{A}{B}\eqdef \ex[a\from A]{\log\frac{\pr{A=a}}{\pr{B=a}}}\;.
	\end{displaymath}
\end{definition}

\begin{definition}[Conditional relative entropy]
	For pairs of random variables $(A, X)$ and $(B,Y)$, and
	$(a,x)\in\Supp(A, X)$,  the \emph{conditional sample relative entropy} is:
	\begin{displaymath}
		\skl{a, x}{A|X}{B|Y} \eqdef \log\frac{\pr{A = a|X = x}}{\pr{B = a|Y = x}}
		\,,
	\end{displaymath}
	and the \emph{conditional relative entropy} is:
	\begin{displaymath}
		\kl{A|X}{B|Y}\eqdef
		\ex[(a, x)\from (A, X)]{\log\frac{\pr{A = a|X = x}}{\pr{B = a|Y = x}}}
		\,.
	\end{displaymath}
\end{definition}

\begin{proposition}[Chain rule for relative entropy]
	\label{prop:kl-chain}
	For pairs of random variables $(X, A)$ and $(Y, B)$:
	\begin{displaymath}
		\kl{A, X}{B, Y} = \kl{A|X}{B|Y} + \kl{X}{Y}
		\,,
	\end{displaymath}
	and for $(a,x)\in\Supp(A, X)$:
	\begin{displaymath}
		\skl{a,x}{A, X}{B, Y} = \skl{a,x}{A|X}{B|Y} + \skl{x}{X}{Y}
		\,.
	\end{displaymath}
\end{proposition}

\begin{proposition}[Data-processing inequality]
	\label{lemma:dec}
	Let $(X, Y)$ be a pair of random variables and let $f$ be a function defined
	on $\Supp(Y)$, then:
	\begin{displaymath}
		\kl{X}{Y}\geq \kl{f(X)}{f(Y)}
		\,.
	\end{displaymath}
\end{proposition}

\begin{definition}[min relative entropy]
	Let $(X, Y)$ be a pair of random variables and $\klfail\in[0,1]$. We define
	$\klm{\klfail}{X}{Y}$ to be the quantile of level $\klfail$ of
	$\skl{x}{X}{Y}$, equivalently it is the smallest $\klgap\in\bbR$
	satisfying:
	\begin{displaymath}
		\pr[x\from X]{\skl{x}{X}{Y}\leq \klgap}\geq \klfail
		\;,
	\end{displaymath}
	and it is characterized by the following equivalence:
	\begin{displaymath}
	    \klm{\klfail}{X}{Y} > \klgap \iff
		\pr[x\from X]{\skl{x}{X}{Y}\leq \klgap}< \klfail
		\;.
	\end{displaymath}
\end{definition}

\paragraph{Block generators}

\begin{definition}[Block generator]
	An \emph{$m$-block generator} is a function
	$\sG:\zo^s\to\prod_{i=1}^m\zo^{\ell_i}$. $\sG_i(r)$ denotes the $i$-th
	block of $\sG$ on input $r$ and $|\sG_i|=\ell_i$ denotes the bit length of
	the $i$-th block.
\end{definition}

\begin{definition}[Online generator]
	An \emph{online $m$-block generator} is a function
	$\tG:\prod_{i=1}^m\zo^{s_i}\to\prod_{i=1}^m\zo^{\ell_i}$ such that for all
	$i\in[m]$ and $r\in\prod_{i=1}^m\zo^{s_i}$, $\tG_i(r)$ only depends on
	$r_{\leq i}$. We sometimes write $\tG_i(r_{\leq i})$ when the input blocks
	$i+1,\dots,m$ are unspecified.
\end{definition}

\begin{definition}[Support]
	The \emph{support} of a generator $\sG$ is the support of the random
	variable $\Supp\big(\sG(R)\big)$ for uniform input $R$. If $\sG$ is an
	$(m+1)$-block generator, and $\Pi$ is a binary relation, we say that $\sG$
	is \emph{supported on $\Pi$} if $\Supp\big(\sG_{\leq m}(R),
	\sG_{m+1}(R)\big)\subseteq \Pi$.
\end{definition}

When $\sG$ is an $(m+1)$-block generator supported on a binary relation $\Pi$,
we will often use the notation $\sG_{\w}\eqdef \sG_{m+1}$ to emphasize that
the last block corresponds to a witness for the first $m$ blocks.

\paragraph{Cryptography.}

\begin{definition}[One-way Function]
	Let $n$ be a security parameter, $t=t(n)$ and $\eps=\eps(n)$.  A function
	$f:\zo^n\to\zo^n$ is a \emph{$(t, \eps)$-one-way function} if:
	\begin{enumerate}
		\item For all time $t$ randomized algorithm $\sA$:
			$
		\pr[x\from U_n]{\sA\big(f(x)\big) \in f^{-1}\big(f(x)\big)} \leq \eps
		$, where $U_n$ is uniform over $\zo^n$.
	\item There exists a polynomial time algorithm $\sB$ such that
		$\sB(x)=f(x)$ for all $x\in\zo^n$.
	\end{enumerate}
	If $f$ is $(n^c,1/n^c)$-one-way for every $c\in\bbN$, we say that $f$ is
	\emph{(strongly) one-way}.
\end{definition}

\section{Search Problems and Hardness in Relative Entropy}

In this section, we first present the classical notion of hard-on-average
search problems and introduce the new notion of hardness in relative entropy.
We then relate the two notions by proving that average-case hardness implies
hardness in relative entropy.

\subsection{Search problems}

For a binary relation $\Pi\subseteq\zo^*\times\zo^*$, we write $\Pi(y, w)$ for
the predicate that is true iff $(y,w)\in\Pi$ and say that $w$ is
a \emph{witness} for the \emph{instance} $y$\footnote{We used the
	unconventional notation $y$ for the instance (instead of $x$) because our
relations will often be of the form $\Pi^f$ for some function $f$; in this case
an instance is some $y$ in the range of $f$ and a witness for $y$ is any
preimage $x\in f^{-1}(y)$.}.  To each relation $\Pi$, we naturally associate
(1) a \emph{search problem}: given $y$, find $w$ such that $\Pi(y,w)$ or state
that no such $w$ exist and (2) the \emph{decision problem} defined by the
language $L_\Pi \eqdef \{y\in\zo^*\,:\, \exists w\in\zo^*,\, \Pi(y,w)\}$.
$\FNP$ denotes the set of all relations $\Pi$ computable by a polynomial time
algorithm and such that there exists a polynomial $p$ such that
$\Pi(y,w)\Rightarrow |w|\leq p(|y|)$.  Whenever $\Pi\in\FNP$, the associated
decision problem $L_\Pi$ is in $\NP$. We now define average-case hardness.

\begin{definition}[distributional search problem]\label{defn:search}
	A \emph{distributional search problem} is a pair $(\Pi, Y)$ where
	$\Pi\subseteq\zo^*\times\zo^*$ is a binary relation and $Y$ is a random
	variable supported on $L_\Pi$.

	The problem $(\Pi, Y)$ is \emph{(t, $\invp$)-hard} if $\Pr\big[\Pi\big(Y,
	\sA(Y)\big)\big]\leq \invp$ for all time $t$ randomized algorithm $\sA$,
	where the probability is over the distribution of $Y$ and the randomness of
	$\sA$.

\end{definition}

\begin{example}
	\label{ex:invf}
	For $f:\zo^n\to\zo^n$, the problem of inverting $f$ is the search problem
	associated with the relation $\Pi^f\eqdef\{(f(x),
	x)\,:\,x\in\zo^n\}$.
	If $f$ is a $(t, \eps)$-one-way function, then the distributional search
	problem $\big(\Pi^f, f(X)\big)$ of inverting $f$ on a uniform random input
	$X\in\zo^n$ is $(t, \eps)$-hard.
\end{example}

\begin{remark}
	Consider a distributional search problem $(\Pi, Y)$. Without loss of
	generality, there exists a (possibly inefficient) two-block generator
	$\sG=(\sG_1,\sG_\w)$
	supported on $\Pi$ such that $\sG_1(R) = Y$ for uniform input $R$. If
	$\sG_\w$ is polynomial-time computable, it is easy to see that the search
	problem $\big(\Pi^{\sG_1}, \sG_1(R)\big)$ is at least as hard as $(\Pi,
	Y)$.  The advantage of writing the problem in this ``functional'' form is that
	the distribution $(\sG_1(R), R)$ over (instance, witness) pairs is flat,
	which is a necessary condition to relate hardness to inaccessible entropy
	(see Theorem~\ref{thm:nb-inacc}).

	Furthermore, if $\sG_1$ is also polynomial-time computable and $(\Pi, Y)$
	is $(\poly(n), \negl(n))$-hard, then $R\mapsto \sG_1(R)$ is a one-way
	function. Combined with the previous example, we see that the existence of
	one-way functions is equivalent to the existence of $(\poly(n),
	\negl(n))$-hard search problems for which (instance, witness) pairs can be
	efficiently sampled.
\end{remark}

\subsection{Hardness in relative entropy}

Instead of considering an adversary directly attempting to solve a search
problem $(\Pi, Y)$, the adversary in the definition of hardness in relative
entropy comprises a pair of algorithm $(\tG, \sS)$ where $\tG$ is a two-block
generator outputting valid (instance, witness) pairs for $\Pi$ and $\sS$ is
a \emph{simulator} for $\tG$: given an instance $y$, the goal of $\sS$ is to
output randomness $r$ for $\tG$ such that $\tG_1(r) = y$. Formally, the
definition is as follows.

\begin{definition}[hardness in relative entropy]
	\label{def:kl-hard}
	Let $(\Pi, Y)$ be a distributional search problem. We say that $(\Pi, Y)$
	has \emph{hardness $(t, \klgap)$ in relative entropy} if:
	\begin{displaymath}
		\kl{\tR, \tG_1(\tR)}{\sS(Y), Y} > \klgap
		\;,
	\end{displaymath}
	for all pairs $(\tG, \sS)$ of time $t$ algorithms where $\tG$ is
	a two-block generator supported on $\Pi$ and $\tR$ is uniform randomness
	for $\tG_1$.  Similarly, for $\klfail\in[0,1]$, $(\Pi, Y)$ has
	\emph{hardness $(t, \klgap)$ in $\klfail$-min relative entropy} if for all
	such pairs:
	\begin{displaymath}
		\klm{\klfail}{\tR, \tG_1(\tR)}{\sS(Y), Y} > \klgap
		\;.
	\end{displaymath}
\end{definition}

Note that a pair $(\tG, \sS)$ achieves a relative entropy of zero in
Definition~\ref{def:kl-hard} if $\tG_1(R)$ has the same distribution as $Y$ and
if $\tG_1\big(\sS(y)\big) = y$ for all $y\in\Supp(Y)$. In this case, writing
$\tG_\w \eqdef \tG_2$, we have that $\tG_\w\big(\sS(Y)\big)$ is a valid witness
for $Y$ since $\tG$ is supported on $\Pi$. 

More generally, the composition $\tG_\w\circ \sS$ solves the search problem
$(\Pi, Y)$ whenever $\tG_1\big(\sS(Y)\big) = Y$.  When the relative entropies
in Definition~\ref{def:kl-hard} are upper-bounded, we can lower bound the
probability of the search problem being solved (Lemma~\ref{lem:kl-hard}). This
immediately implies that hard search problems are also hard in relative
entropy.

\begin{theorem}
	\label{thm:kl-hard}
	Let $(\Pi, Y)$ be a distributional search problem. If $(\Pi, Y)$ is $(t,
	\invp)$-hard, then it has hardness $(t', \klgap')$ in relative entropy and $(t',
	\klgap'')$ in $\klfail$-min relative entropy for every
	$\klfail\in[0,1]$ where $t' = \Omega(t)$,\footnote{For the theorems in this
		paper that relate two notions of hardness, the notation $t'=\Omega(t)$
	means that there exists a constant $C$ depending \emph{only} on the
computational model such that $t'\geq C\cdot t$.} $\klgap' = \log(1/\eps)$ and
$\klgap'' = \log(1/\invp)-\log(1/\klfail)$.
\end{theorem}

\begin{remark}
	As we see, a ``good'' simulator $\sS$ for a generator $\tG$ is one for
	which $\tG_1\big(\sS(Y)\big) = Y$ holds often. It will be useful in
	Section~\ref{sec:iekl} to consider simulators $\sS$ which are allowed to
	fail by outputting a failure string $r\notin\Supp(\tR)$, (e.g.~$r=\bot$)
	and adopt the convention that $\tG_1(r)=\bot$ whenever $r\notin\Supp(\tR)$.
	With this convention, we can without loss of generality add the requirement
	that $\tG_1\big(\sS(Y)\big)=Y$ whenever $\sS(Y)\in\Supp(\tR)$: indeed,
	$\sS$ can always check that it is the case and if not output a failure
	symbol. For such a simulator $\sS$, observe that  for all $r\in\Supp(\tR)$,
	the second variable on both sides of the relative entropy in
	Definition~\ref{def:kl-hard} is obtained by applying $\tG_1$ on the first
	variable and can thus be dropped, leading to a simpler definition of
	hardness in relative entropy: $\kl{\tR}{\sS(Y)} > \klgap$.
\end{remark}

Theorem~\ref{thm:kl-hard} is an immediate consequence of the following lemma.

\begin{lemma}
	\label{lem:kl-hard}
	Let $(\Pi, Y)$ be a distributional search problem and $(\tG, \sS$) be
	a pair of algorithms with $\tG=(\tG_1,\tG_\w)$ a two-block generator
	supported on $\Pi$.  Define the linear-time oracle algorithm
	$\sA^{\tG_\w, \sS}(y) \eqdef \tG_\w(\sS(y))$.
	For $\klgap\in\bbR^+$ and
	$\klfail\in[0,1]$:
	\begin{enumerate}
		\item If $\kl{\tR, \tG_1(\tR)}{\sS(Y), Y}\leq\klgap$ then $\pr{\Pi(Y,
			\sA^{\tG_\w, \sS}(Y))}\geq 1/2^\klgap$.
		\item If $\klm{\klfail}{\tR, \tG_1(\tR)}{\sS(Y), Y}\leq\klgap$ then
			$\pr{\Pi(Y, \sA^{\tG_\w, \sS}(Y))}\geq \klfail/2^\klgap$.
	\end{enumerate}
\end{lemma}

\begin{proof}
	We have:
	\begin{align*}
		\pr{\Pi\big(Y, \sA^{\tG_\w, \sS}(Y)\big)}
		&= \pr{\Pi(Y, \tG_\w(\sS(Y)))}\\
		&\geq \pr{\tG_1(\sS(Y)) = Y}\tag{$\tG$ is supported on $\Pi$}\\
		&= \sum_{r\in\Supp(\tR)} \pr{\sS(Y)=r\wedge Y=\tG_1(r)}\\
		&=\ex[r\from\tR]{\frac{\pr{\sS(Y)=r\wedge Y=\tG_1(r)}}{\pr{\tR=r}}}\\
		&= \ex[\substack{r\from \tR\\y\from \tG_1(r)}] {2^{-\skl{r, y}{\tR,
		\tG_1(\tR)}{\sS(Y), Y}}}
		\,.
	\end{align*}
	Now, the first claim follows by Jensen's inequality (since $x\mapsto
	2^{-x}$ is convex) and the second claim follows by Markov' inequality when
	considering the event that the sample relative entropy is smaller than
	$\klgap$ (which occurs with probability at least $\klfail$ by assumption).
\end{proof}

\paragraph{Relation to relative pseudoentropy.}

In \cite{VZ12}, the authors introduced the notion of relative
pseudoentropy\footnote{As already mentioned in the introduction, this notion
was in fact called ``$\KL$-hardness for sampling'' in \cite{VZ12} but we rename
it here to unify the terminology between the various notions discussed here.}:
for jointly distributed variables $(Y, W)$, $W$ has relative pseudoentropy given $Y$
if it is hard for a polynomial time adversary to approximate---measured in
relative entropy---the conditional distribution $W$ given $Y$.  Formally:

\begin{definition}[relative pseudoentropy, Def.~3.4 in \cite{VZ12}]
	Let $(Y, W)$ be a pair of random variables, we say that $W$ has
	\emph{relative pseudoentropy $(t, \klgap)$ given $Y$} if for all time $t$
	randomized algorithm $\sS$, we have:
	\begin{displaymath}
		\kl{Y, W}{Y, \sS(Y)}>\klgap
		\,.
	\end{displaymath}
\end{definition}

As discussed in Section~\ref{sec:nbpe}, it was shown in \cite{VZ12} that if
$f:\zo^n\to\zo^n$ is a one-way function, then $\big(f(X), X_1,\dots,X_n)$ has
	next-bit pseudoentropy for uniform $X\in\zo^n$ (see
	Theorem~\ref{thm:VZ-intro}). The first step in proving this result was to
	prove that $X$ has relative pseudoentropy given $f(X)$
	(see Lemma~\ref{lem:OWF-KL-hard-to-sample-intro}).

We observe that when $(Y, W)$ is of the form $(f(X), X)$ for some function
$f:\zo^n\to\zo^n$ and variable $X$ over $\zo^n$, then relative pseudoentropy is
implied by hardness in relative entropy by simply fixing $\tG$ to be the
``honest sampler'' $\tG(X) = (f(X), X)$. Indeed, in this case we have:
\begin{displaymath}
	\kl{X, \tG_1(X)}{\sS(Y), Y} = \kl{X, f(X)}{\sS(Y), Y}
	\,.
\end{displaymath}
We can thus recover Lemma~\ref{lem:OWF-KL-hard-to-sample-intro} as a direct
corollary of Theorem~\ref{thm:kl-hard}.

\begin{corollary}
	Consider a function $f:\zo^n\to\zo^n$ and define $\Pi^f \eqdef \{(f(x),
	x):x\in\zo^n\}$ and $Y \eqdef f(X)$ for $X$ uniform over $\zo^n$. If $f$ is
	$(t,\eps)$-one-way, then $(\Pi^f,Y)$ has hardness
	$\big(t',\log(1/\eps)\big)$ in relative entropy and $X$ has relative
	pseudoentropy $\big(t',\log(1/\eps)\big)$ given $Y$ with $t'=\Omega(t)$.
\end{corollary}

\paragraph{Witness hardness in relative entropy.}

We also introduce a relaxed notion of hardness in relative entropy called
witness hardness in relative entropy.  In this notion, we further require
$(\tG, \sS)$ to approximate the joint distribution of (instance, witness) pairs
rather than only instances.  For example, the problem of inverting
a function $f$ over a random input $X$ is naturally associated with the
distribution $\big(f(X), X\big)$.  The relaxation in this case is analogous to
the notion of \emph{distributional one-way function} for which the adversary is
required to approximate the uniform distribution over preimages.

\begin{definition}[witness hardness in relative entropy]
	\label{def:dist-hard}
	Let $\Pi$ be a binary relation and $(Y, W)$ be a pair of random variables
	supported on $\Pi$. We say that $(\Pi, Y, W)$ has \emph{witness hardness $(t, \klgap)$
	in relative entropy} if for all pairs of time $t$
	algorithms $(\tG, \sS)$ where $\tG$ is a two-block generator supported on
	$\Pi$, for uniform $\tR$:
	\begin{displaymath}
		\kl{\tR, \tG_1(\tR), \tG_\w(\tR)}{\sS(Y), Y, W} > \klgap\;.
	\end{displaymath}
	Similarly, for $\klfail\in[0,1]$, $(\Pi, Y, W)$ has \emph{witness hardness
	$(t, \klgap)$ in $\klfail$-min relative entropy}, if for all such pairs:
	\begin{displaymath}
		\klm{\klfail}{\tR, \tG_1(\tR),
		\tG_\w(\tR)}{\sS(Y), Y, W} > \klgap\;.
	\end{displaymath}
\end{definition}

We introduced hardness in relative entropy first, since it is the notion which
is most directly obtained from the hardness of distribution search problems.
Observe that by the data processing inequality for relative entropy
(Proposition~\ref{lemma:dec}), dropping the third variable on both sides of the
relative entropies in Definition~\ref{def:dist-hard} only decreases them.
Hence, hardness in relative entropy implies witness hardness as stated in
(Theorem~\ref{thm:witness-kl-hard}). As we will see in Section~\ref{sec:iekl}
witness hardness in relative entropy is the ``correct'' notion to obtain
inaccessible entropy from: it is in fact equal to inaccessible entropy up to
$1/\poly$ losses.

\begin{theorem}
	\label{thm:witness-kl-hard}
	Let $\Pi$ be a binary relation  and $(Y, W)$ be a pair of random variables
	supported on $\Pi$.  If $(\Pi, Y)$ is $(t, \eps)$-hard, then $(\Pi, Y, W)$
	has witness hardness $(t', \klgap')$ in relative entropy and $(t',
	\klgap'')$ in $\klfail$-min relative entropy for every
	$\klfail\in[0,1]$ where $t' = \Omega(t)$, $\klgap' = \log(1/\eps)$ and
	$\klgap'' = \log(1/\invp)-\log(1/\klfail)$.
\end{theorem}

\begin{remark}
	The data processing inequality does not hold exactly for $\KLM$, hence the
	statement about $\klfail$-min relative entropy in
	Theorem~\ref{thm:witness-kl-hard} does not follow with the claimed
	parameters in a black-box manner from Theorem~\ref{thm:kl-hard}. However,
	an essentially identical proof given in Appendix~\ref{app:proof} yields the
	result.
\end{remark}

\section{Inaccessible Entropy and Hardness in Relative Entropy}
\label{sec:iekl}

In this section, we relate our notion of witness hardness in relative entropy
to the inaccessible entropy definition of \cite{HRVW16}.  Roughly speaking, we
``split'' the relative entropy into blocks and obtain the intermediate notion
of next-block inaccessible relative entropy (Section~\ref{sec:rs}) which we then
relate to inaccessible entropy (Section~\ref{sec:ie}). Together, these results
show that if $f$ is a one-way function, the generator $\sG^f(X)
= \big(f(X)_1,\dots,f(X)_n, X\big)$ has superlogarithmic inaccessible entropy.

\subsection{Next-block hardness and rejection sampling}
\label{sec:rs}

For an online (adversarial) generator $\tG$, it is natural to consider
simulators $\sS$ that also operate in an online fashion. That is:

\begin{definition}[online simulator]
	Let $\tG:\prod_{i=1}^m\zo^{s_i}\to\prod_{i=1}^m\zo^{\ell_i}$ be an online
	$m$-block generator. An \emph{online simulator} for $\tG$ is a PPT
	algorithm $\sS$ such that for all
	$y=(y_1,\dots,y_m)\in\prod_{i=1}^m\zo^{\ell_i}$, defining inductively
	$\hr_i \eqdef \sS(\hr_{<i}, y_i)\in\zo^{s_i}$, we have for all $i\in[m]$:
	\begin{displaymath}
		\tG_i(\hr_{\leq i}) = y_i \quad\mathrm{or}\quad  \hr_i = \bot\,.
	\end{displaymath}

		The \emph{running time} of $\sS$ is the total amount of time required
		to compute $\hr_1,\dots,\hr_m$.
\end{definition}

The goal of such an online simulator $\sS$ is to ensure that the distribution
of $\hR_i = \sS( \hr_{<i}, y_i)$ is close to that of $\tR_i
| \tR_{<i}=\hr_{<i}, \tY_i = y_i$ where $(\tY_1,\dots, \tY_m)\eqdef
\tG(\tR_{\leq m})$ for uniformly random $(\tR_1,\dots, \tR_m)$. Equivalently,
$\hR_i$ should be close to uniform on $\{\hr_i\,:\, \tG_i(\hr_{\leq i})=y_i\}$.
Measuring closeness with relative entropy, we have:

%Consider a binary relation $\Pi$ and a pair of random variables  $(Y, W)$
%supported on $\Pi$.  Let $\tG$ be an online $(m+1)$-block generator supported
%on $\Pi$.
%and write $\tY_{\leq m} \eqdef \tG(\tR_{\leq m})$ for uniform
%$\tR_{\leq m}$. For such a generator $\tG$, it is natural to consider
%simulators operating in an online manner.  Specifically, an online simulator in
%this context is a PPT algorithm $\sS$ such that on input $(\hR_{<i}, Y_i)$,
%$\sS$ outputs $\hR_i$ of the same length as $\tR_i$.  The
%goal of $\sS$ is to output random coins such that $(\tR_i, \tY_i)$ is ``close''
%to $(\hR_i, Y_i)$ conditioned on the past. This leads to the following natural
%blockwise notion of hardness in relative entropy for online generators and
%simulators.

\begin{definition}[next-block hardness in relative entropy]
	\label{def:nb-hire}
	The joint distribution $Y=(Y_1, \dots, Y_{m})$ has \emph{next-block
	hardness $(t, \klgap)$ in relative entropy} if the following holds for
	every time $t$ online $m$-block generator $\tG$ and every time $t$ online
	simulator $\sS$ for $\tG$.

	Write $\tY_{\leq m}\eqdef\tG(\tR_{\leq m})$ for uniform $\tR_{\leq m}$, and
	define inductively $\hR_i \eqdef \sS(\hR_{<i}, Y_i)$. Then we require:
	\begin{displaymath}
		\sum_{i=1}^{m} \kl{\tR_i, \tY_i|\tR_{<i}, \tY_{<i}}{\hR_i, Y_i|\hR_{<i},
		Y_{<i}}> \klgap
		\,.
	\end{displaymath}

	Similarly, for $\klfail\in[0,1]$, we say that $(Y_1, \dots, Y_{m})$ has
	\emph{next-block hardness $(t,\klgap)$ in $\klfail$-min relative entropy}
	if, with the same notations as above:
	\begin{displaymath}
		\pr[\substack{r_{\leq m}\from\tR_{\leq m}\\y_{\leq m}\from
			\tG(r_{\leq m})}]{\sum_{i=1}^{m}
		\skl{y_i,r_{<i},y_{<i}}{\tR_i,\tY_i|\tR_{<i}, \tY_{<i}}{\hR_i, Y_i|\hR_{<i},
		Y_{<i}} \leq \klgap} < \klfail
		\,.
	\end{displaymath}
\end{definition}

Observe that using the chain rule for relative entropy, the sum of relative
entropies appearing in Definition~\ref{def:nb-hire} is exactly equal to the
relative entropies appearing in Definition~\ref{def:kl-hard}. Since,
furthermore considering an online generator $\tG$ and online simulator $\sS$
is only less general than arbitrary pairs $(\tG, \sS)$, we immediately obtain
the following theorem.

\begin{theorem}
	\label{thm:hire-nb-hire}
	Let $(\Pi, Y)$ be a distributional search problem. If $(\Pi, Y)$ has
	hardness $(t, \klgap)$ in relative entropy then $(Y_1,\dots, Y_m)$ has
	next-block hardness $(t, \klgap)$ in relative entropy.

	Similarly, for any $\klfail \in [0,1]$, if $(\Pi, Y)$ has hardness $(t,
	\klgap)$ in $\klfail$-min relative entropy then $(Y_1,\dots, Y_m)$ has
	next-block hardness $(t, \klgap)$ in $\klfail$-min relative entropy.
\end{theorem}

\begin{proof}
	Immediate using the chain rule for relative (sample) entropy.
\end{proof}

The next step is to obtain a notion of hardness that makes no reference to
simulators by considering, for an online block generator $\tG$, a specific
simulator $\Sim^{\tG, T}$ which on input
$(\hr_{<i}, y_i)$, generates $\hR_i$ using rejection
sampling until $\tG_i(\hr_{<i}, \hR_{i})=y_i$.  The superscript $T$ is the
maximum number of attempts after which $\Sim^{\tG, T}$ gives up and outputs
$\bot$.  The formal definition of $\Sim^{\tG, T}$ is given in
Algorithm~\ref{algo:bounded}.

\begin{algorithm}
	\caption{Rejection sampling simulator $\Sim^{\tG, T}$ for $1\leq i\leq m$}
	\label{algo:bounded}
	\begin{algorithmic}
		\Require $y_i\in\zo^*$, $\hr_{<i}\in (\zo^v\cup\{\bot\})^{i-1}$
	\Ensure $\hr_i\in\zo^v\cup\{\bot\}$
	\If{$\hr_{i-1} = \bot$}
		\State $\hr_i\gets\bot$; \Return
	\EndIf
	\Repeat
		\State sample $\hr_i\gets\zo^v$
	\Until{$\tG_i(\hr_{\leq i}) = y_i$ or $\geq T$ attempts}
	\If{$\tG_i{(\hr_{\leq i})\neq y_i}$}
		\State $\hr_i \gets \bot$ 
	\EndIf
	\end{algorithmic}
\end{algorithm}

For the rejection sampling simulator $\Sim^{\tG, T}$, we will show
in Lemma~\ref{lem:bkl-reduction} that the next-block
hardness in relative entropy in Definition~\ref{def:nb-hire} decomposes
as the sum of two terms:
\begin{enumerate}
	\item A term measuring how well $\tG_{\leq m}$ approximates the
		distribution $Y$ in an online manner, without any reference to a
		simulator.
	\item An error term measuring the failure probability of the rejection
		sampling procedure due to having a finite time bound $T$.
\end{enumerate}
As we show in Lemma~\ref{lem:stexpectation}, the error term can be made
arbitrarily small by setting the number of trials $T$ in $\Sim^{\tG, T}$
to be a large enough multiple of $m\cdot 2^{\ell}$ where $\ell$ is the
length of the blocks of $\tG_{\leq m}$.
This leads to a $\poly(m)$ time algorithm whenever $\ell$ is logarithmic
in $m$. That is, given an online block generator $\tG$ for which
$\tG_{\leq m}$ has short blocks, we obtain a corresponding simulator
``for free''. Thus, considering only the first term leads to the
following clean definition of next-block inaccessible relative entropy
that makes no reference to simulators.

\begin{definition}[next-block inaccessible relative entropy]
	\label{def:block-hard}
	The joint distribution $(Y_1, \dots, Y_{m})$ has \emph{next-block
	inaccessible relative entropy $(t, \klgap)$}, if for every time $t$ online
	$m$-block generator $\tG$ supported on $Y_{\leq m}$, writing $\tY_{\leq
	m}\eqdef\tG(\tR_{\leq m})$ for uniform $\tR_{\leq m}$, we have:
	\begin{displaymath}
		\sum_{i=1}^{m} \kl{\tY_i|\tR_{<i}, \tY_{<i}}{Y_i|R_{<i},
		Y_{<i}}> \klgap
		\,,
	\end{displaymath}
	where $R_i$ is a ``dummy'' random variable over  the domain of $\tG_i$ and
	independent of $Y_{\leq m+1}$.  Similarly, for $\klfail\in[0,1]$, we say
	that $(Y_1, \dots, Y_{m+1})$ has \emph{next-block inaccessible
	$\klfail$-min relative entropy $(t,\klgap)$} if for every $\tG$ as above:
	\begin{displaymath}
		\pr[\substack{r_{\leq m}\from\tR_{\leq m}\\y_{\leq m}\from
			\tG(r_{\leq m})}]{\sum_{i=1}^{m}
		\skl{y_i,r_{<i},y_{<i}}{\tY_i|\tR_{<i}, \tY_{<i}}{Y_i|R_{<i},
		Y_{<i}} \leq \klgap} < \klfail
		\,,
	\end{displaymath}
	where $(\tY_{\leq m},\tR_{\leq m})$ are defined as above.
\end{definition}

\begin{remark}
	Since $\tY_{<i}$ is a function of $\tR_{<i}$, the first
	conditional distribution in the KL is effectively $\tY_i|\tR_{<i}$.
	Similarly the second distribution is effectively $Y_i|Y_{<i}$. The extra
	random variables are there for syntactic consistency.
\end{remark}

With this definition in hand, we can make formal the claim that, even as
sample notions, the next-block hardness in relative entropy decomposes
as next-block inaccessible relative entropy plus an error term.
\begin{lemma}
	\label{lem:bkl-reduction}
	For a joint distribution $(Y_1, \dots, Y_{m})$, let $\tG$ be an online
	$m$-block generator supported on $Y_{\leq m}$. Define $(\tY_1,\dots,
	\tY_{m})\eqdef \tG(\tR)$ for uniform random variable $\tR = (\tR_1,
	\dots, \tR_{m})$ and let $R_i$ be a ``dummy'' random variable over the
	domain of $\tG_i$ and independent of $Y_{\leq m}$. We also define $\hR_i
	\eqdef \Sim^{\tG, T}(\hR_{<i}, Y_i)$ and $\hY_i = \tG(\hR_{\leq i})$.
	Then, for all $r\in\Supp(\tR)$ and $y\eqdef \tG(r)$:
	\begin{align*}
		\sum_{i=1}^m&\skl{r,y}{\tR_i, \tY_i | \tR_{<i}, \tY_{<i}}{\hR_i, Y_i|\hR_{<i},Y_{<i}}\\
					&= \sum_{i=1}^m \skl{r,y}{\tY_i | \tR_{<i},\tY_{<i}}{Y_i|R_{<i},Y_{<i}}
					+\sum_{i=1}^m\log\left(\frac{1}{\pr{\hY_i = y_i | Y_i=y_i, \hR_{<i} = r_{<i}}}\right)
		\,.
	\end{align*}
	Moreover, the running time of $\Sim^{\tG, T}$ on input $\hR_{< i},
	Y_i$ is $O(\lvert r_i\rvert \cdot  T)$, with at most $T$ oracle calls to
	$\tG$.
\end{lemma}

\begin{proof}
	Consider $r\in\Supp(\tR)$ and $y\eqdef \tG(r)$. Then:
	  \begin{align*}
		  \sum_{i=1}^m&\skl{r,y}{\!\tR_i, \tY_i | \tR_{<i}, \tY_{<i}}{\hR_i, Y_i|\hR_{<i},Y_{<i}}\\
					  &=\sum_{i=1}^m\skl{r,y}{\tR_i, \tY_i | \tR_{<i}, \tY_{<i}}{\hR_i, \hY_i|\hR_{<i},\hY_{<i}}\\
	      &= \sum_{i=1}^m \left(\skl{r,y}{\tR_i|\tR_{<i}, \tY_{\leq i}}{\hR_i|\hR_{<i}, \hY_{\leq i}}
	       + \skl{r,y}{\tY_i | \tR_{<i}, \tY_{<i}}{\hY_i|\hR_{<i},\hY_{<i}}\right)\\
	      &= \sum_{i=1}^m \skl{r,y}{\tY_i | \tR_{<i}, \tY_{<i}}{\hY_i|\hR_{<i},\hY_{<i}}
	      = \sum_{i=1}^m \skl{r,y}{\tY_i | \tR_{<i}}{\hY_i|\hR_{<i}}
		\,.
	  \end{align*}
	  The first equality is because $Y_i = \hY_i$ since we are only considering
	  non-failure cases ($r_i\neq \bot$). The second equality is the chain
	  rule.  The penultimate equality is by definition of rejection
	sampling: $\tR_i|\tR_{<i},\tY_{\leq i}$ and $\hR_i|\hR_{<i},\hY_{\leq i}$
	are identical on $\Supp(\tR_i)$ since conditioning on $\hY_i = y$ implies
	that only non-failure cases ($r_i\neq\bot$) are considered.
	The last equality is because $\tY_{<i}$ (resp.  $\hY_{<i}$) is
	a deterministic function of $\tR_{<i}$ (resp.  $\hR_{<i}$).

	We now relate $\hY_i|\hR_{<i}$  to $Y_i|Y_{<i}$:
	\begin{align*}
		&\pr{\hY_i = y_i | \hR_{<i} = r_{<i}}
		=\pr{\hY_i = y_i, Y_i=y_i | \hR_{<i} = r_{<i}}
		\tag{$\hY_i=y_i\Leftrightarrow \hY_i=y_i\wedge Y_i=y_i$}\\
		&\quad\,\,=\pr{\hY_i = y_i | Y_i=y_i, \hR_{<i} = r_{<i}}
		\cdot\pr{Y_i = y_i | \hR_{<i} = r_{<i}}
		\tag{Bayes' Rule}\\
		&\quad\,\,=\pr{\hY_i = y_i | Y_i=y_i, \hR_{<i} = r_{<i}}
		\cdot\pr{Y_i = y_i | Y_{<i} = y_{<i}}\,,
	\end{align*}
	where the last equality is because when $r\in\Supp(\tR)$,
	$\hR_{<i}=r_{<i}\Rightarrow Y_{<i}=y_{<i}$ and because $Y_i$ is independent
	of $\hR_{<i}$ given $Y_{<i}$ (as $\hR_{<i}$ is simply a randomized function
	of $Y_{<i}$). The conclusion of the lemma follows by combining the previous
	two derivations.
\end{proof}

Observe that taking expectations with respect to a uniform $\tR$ on both sides
in the conclusion of Lemma~\ref{lem:bkl-reduction}, we get that next-block
hardness in relative entropy is equal to the sum of next-block inaccessible
relative entropy and the expectation of the error term coming from the
rejection sampling procedure. The following lemma upper bounds this
expectation.

\begin{lemma}\label{lem:stexpectation}
	Let $\tG$ be an online $m$-block generator, and let $L_i\eqdef 2^{|\tG_i|}$
	be the size of the codomain of $\tG_i$, $i\in[m]$. Then for all $i\in[m]$,
	$r_{<i}\in\Supp(\tR_{<i})$ and uniform $\tR_i$:
	\begin{displaymath}
		\ex[y_i\from\tG_i(r_{<i},\tR_i)]{\log \frac{1}{\pr{\hY_i=y_i|Y_{i} = y_{i},\hR_{<i} = r_{<i}}}} \leq
		\log\left(1+\frac{L_i-1}{T}\right)\;.
	\end{displaymath}
\end{lemma}

\begin{proof}[Proof of Lemma~\ref{lem:stexpectation}]
	By definition of $\Sim^{\tG, T}$, we have:
	\begin{displaymath}
		\pr{\hY_i=y_i|Y_{i} = y_{i},\hR_{<i} = r_{<i}}
		= 1 - \left(1-\pr{\tG_i(r_{<i}, \tR_i) = y_i}\right)^T
		\,.
	\end{displaymath}
	Applying Jensen's inequality, we have:
	\begin{align*}
		&\ex[y_i\from\tG_i(r_{<i},\tR_i)]{\log \left(\frac{1}{\pr{\hY_i=y_i|Y_{i} = y_{i},\hR_{<i} = r_{<i}}}\right)}\\
		&\qquad\qquad\qquad\leq \log \ex[y_i\from\tG_i(r_{<i},\tR_i)]{\frac{1}{\pr{\hY_i=y_i|Y_{i}
		= y_{i},\hR_{<i} = r_{<i}}}}\\
		&\qquad\qquad\qquad=\log\left(\sum_{y\in\mathrm{Im}(\tG_i(r_{<i}, \cdot))}\frac{p_y}{1-(1-p_y)^T}\right)
	\end{align*}
	where $p_y=\pr{\tG_i(r_{<i}, \tR_i) = y}$.
	Since the function $x/\left(1-(1-x)^T\right)$ is convex (see
	Lemma~\ref{claim:convex} in the appendix), the maximum of the expression
	inside the logarithm over probability distributions $\{p_y\}$ is achieved
	at the extremal points of the standard probability simplex.
	Namely, when all but one $p_y\to 0$ and the other one is $1$.
	Since $\lim_{x\to 0}x/\big(1-(1-x)^T\big) = 1/T$:
	\begin{displaymath}
		\log\left(\sum_{y\in\mathrm{\rm Im}(\tG_i)}\frac{p_y}{1-(1-p_y)^T}\right)
		\leq \log\left(1+ (L_i-1)\cdot \frac{1}{T}\right)\,.\qedhere
	\end{displaymath}
\end{proof}

By combining Lemmas~\ref{lem:bkl-reduction} and \ref{lem:stexpectation}, we are
now ready to state the main result of this section, relating witness hardness
in relative entropy to next-block inaccessible relative entropy.

\begin{theorem}
	\label{thm:bkl-hard}
	Let $\Pi$  be a binary relation and let $(Y, W)$ be a pair of random
	variables supported on $\Pi$.  Let $Y = (Y_1, \dots, Y_m)$ be a partition
	of $Y$ into blocks of at most $\ell$ bits. Then we have:
	\begin{enumerate}
		\item if $(\Pi, Y, W)$ has witness hardness $(t, \klgap)$ in relative
			entropy, then for every $0< \klgap' \leq \klgap$, $(Y_1, \dots,
			Y_m, W)$ has  next-block inaccessible relative entropy $(t',
			\klgap-\klgap')$ where $t' = \Omega(t\klgap'/(m^22^\ell))$.
		\item if $(\Pi, Y, W)$ has witness hardness $(t, \klgap)$ in
			$\klfail$-min relative entropy then for every $0 < \klgap' \leq
			\klgap$ and $0 \leq \klfail' \leq 1-\klfail$, we have that $(Y_1,
			\dots, Y_m, W)$ has  next-block inaccessible
			$(\klfail+\klfail')$-min relative entropy $(t', \klgap-\klgap')$
			where $t' = \Omega(t\klfail'\klgap'/(m^22^\ell))$.
	\end{enumerate}
\end{theorem}

\begin{proof}
	We consider an online generator $\tG$ supported on $(Y_1,\dots, Y_m, W)$
	and the simulator $\Sim^{\tG, T}$. For
	convenience, we sometimes write $Y_{m+1}$ for $W$. Define $\tR\eqdef
	\tR_{\leq m}$ where $\tR_{\leq m}$ is a sequence of independent and
	uniformly random variables, $\tY_{\leq m+1} \eqdef\tG(\tR)$, $\tG_1(\tR)
	\eqdef \tY_{\leq m}$ and $\tG_\w(\tR) \eqdef \tY_{m+1}$. We also write for
	$1\leq i\leq m$, $\hR_i \eqdef \Sim^{\tG, T}(\hR_{<i}, Y_i)$, $\hY_i
	\eqdef \tG(\hR_{\leq i})_i$. Finally we define $\sS^{\tG, T}(Y) \eqdef
	\hR_{\leq m}$.

	Observe that $(\tG_1, \tG_\w)$ is a two-block generator supported on $\Pi$,
	so the pair $(\tG, \sS^{\tG, T})$ forms a pair a algorithms as in the
	definition of witness hardness in relative entropy
	(Definition~\ref{def:dist-hard}). We focus on sample notions first, and
	consider $r\in\Supp(\tR)$, $y\in\Supp(\tY_{\leq m})$ and
	$w\in\Supp(\tY_{m+1})$. First we use the chain rule to isolate the witness
	block:
	\begin{align*}
		&\skl{r,y, w}{\tR, \tG_1(\tR), \tG_\w(\tR)}{\sS^{\tG, T}(Y), Y, W}\\
		&\quad\quad= \skl{r,y, w}{\tG_\w(\tR)|\tR, \tG_1(\tR)}{W | \sS^{\tG, T}(Y),
		Y}
		+\skl{r,y, w}{\tR, \tG_1(\tR)}{\sS^{\tG, T}(Y), Y}\\
		&\quad\quad= \skl{r,y, w}{\tY_{m+1}|\tR_{\leq m}, \tY_{\leq m}}{Y_{m+1}
		| R_{\leq m}, Y_{\leq m}}
		+\skl{r,y, w}{\tR, \tG_1(\tR)}{\sS^{\tG, T}(Y), Y}
		\,.
	\end{align*}

	Next, as in Theorem~\ref{thm:hire-nb-hire} we apply the chain rule to
	decompose the second term on the right-hand side and obtain next-block
	hardness in relative entropy:
	\begin{displaymath}
		\skl{r,y, w}{\tR, \tG_1(\tR)}{\sS^{\tG, T}(Y), Y}
		= \sum_{i=1}^{m} \skl{r,y,w}{\tR_i, \tY_i|\tR_{<i}, \tY_{<i}}{\hR_i, Y_i|\hR_{<i},
		Y_{<i}}
		\,.
	\end{displaymath}

	Finally, we use Lemma~\ref{lem:bkl-reduction} to further decompose the
	right-hand side term into inaccessible relative entropy and the rejection
	sampling error:
	\begin{align*}
		 &\sum_{i=1}^{m} \skl{r,y,w}{\tR_i, \tY_i|\tR_{<i}, \tY_{<i}}{\hR_i, Y_i|\hR_{<i},
		Y_{<i}}\\
		 &\quad=\sum_{i=1}^m \skl{r,y}{\tY_i
		| \tR_{<i},\tY_{<i}}{Y_i|R_{<i},Y_{<i}}
		+\sum_{i=1}^m\log\left(\frac{1}{\pr{\hY_i = y_i | Y_i=y_i, \hR_{<i} = r_{<i}}}\right)
		\,.
	\end{align*}

	Combining the previous derivations, we obtain:
	\begin{align*}
		&\skl{r,y, w}{\tR, \tG_1(\tR), \tG_\w(\tR)}{\sS^{\tG, T}(Y), Y, W}\\
		&\quad=\sum_{i=1}^{m+1} \skl{r,y}{\tY_i | \tR_{<i},\tY_{<i}}{Y_i|R_{<i},Y_{<i}}
		+\sum_{i=1}^m\log\left(\frac{1}{\pr{\hY_i = y_i | Y_i=y_i, \hR_{<i} = r_{<i}}}\right)
		\,.
	\end{align*}

	Now, the first claim of the theorem follows by taking expectations on both
	sides and observing that when $T = m\cdot 2^{\ell}/(\klgap'\ln 2)$,
	Lemma~\ref{lem:stexpectation} implies that the expected value of the
	rejection sampling error is smaller than $\klgap'$. 

	For the second claim, we first establish  using
	Lemma~\ref{lem:stexpectation} and Markov's inequality that:
	\begin{displaymath}
		\pr[\substack{y_{\leq m+1}\from\tY_{\leq m+1}\\ r\from \tR}] {
		\sum_{i = 1}^m \log\left(\frac{1}{\pr{\hY_i=y_i|\hR_{<i} = r_{<i}, \hY_{<i} = y_{<i}}}\right)
		\geq \frac{m\cdot 2^\ell}{T\klfail'\ln 2}
	}
	\leq \klfail'
	\end{displaymath}
	and we reach a similar conclusion by setting $T = m\cdot
	2^{\ell}/(\klfail'\klgap'\ln 2)$.
\end{proof}

\begin{remark}
	For fixed distribution and generators, in the limit where $T$ grows to
	infinity, the error term caused by the failure of rejection sampling in
	time $T$ vanishes. In this case, hardness in relative entropy implies
	next-block inaccessible relative entropy without any loss in the hardness
	parameters.
\end{remark}

\subsection{Next-block inaccessible relative entropy and inaccessible entropy}
\label{sec:ie}

We first recall the definition from \cite{HRVW16}, slightly adapted to our
notations.

\begin{definition}[Inaccessible Entropy]
	\label{def:ie}
	Let $(Y_1,\dots,Y_{m+1})$ be a joint distribution.\footnote{We write $m+1$
		the total number of blocks, since in this section we will think of
		$Y_{m+1}$ (also written as $W$) as the witness of a distributional
		search problem and $(Y_1,\dots,Y_m)$ are the blocks of the instance as
		in the previous section.}
	We say that $(Y_1,\dots,Y_{m+1})$ has \emph{inaccessible entropy
	$(t,\klgap)$} if for all $(m+1)$-block online generators $\tG$ running in
	time $t$ and consistent with $(Y_1,\dots,Y_{m+1})$:
	\begin{displaymath}
		\sum_{i=1}^{m+1} \left(\Ent(Y_i|Y_{<i}) - \Ent(\tY_i|\tR_{< i}) \right)> \klgap
		\;.
	\end{displaymath}
	where $(\tY_1,\dots,\tY_{m+1}) = \tG(\tR_1,\dots,\tR_{m+1})$ for a uniform
	$\tR_{\leq m+1}$.

	Similarly $(Y_1, \dots, Y_{m+1})$ has \emph{inaccessible $\klfail$-max
	entropy $(t,\klgap)$} if for all $(m+1)$-block online generators $\tG$
	running in time $t$ and  consistent with $(Y_1,\dots,Y_{m+1})$:
	\begin{displaymath}
		\pr[\substack{r_{\leq m+1}\from \tR_{\leq m+1}\\y_{\leq m+1}\from\tG(r_{\leq m+1})}]{
		\sum_{i=1}^{m+1} \left(\sent{y_i,y_{<i}}{Y_i|Y_{<i}} - \sent{y_i,r_{<i}}{\tY_i|\tR_{< i}} \right)\leq \klgap
		}
		< \klfail
		\;.
	\end{displaymath}
\end{definition}

Unfortunately, one unsatisfactory aspect of Definition~\ref{def:ie} is that
inaccessible entropy can be negative since the generator $\tG$ could have more
entropy than $(Y_1,\dotsc,Y_{m+1})$: if all the $Y_i$ are independent biased
random bits, then a generator $\tG$ outputting unbiased random bits will have
negative inaccessible entropy. On the other hand, next-block inaccessible
relative entropy (Definition~\ref{def:block-hard}) does not suffer from this
drawback.

Moreover, in the specific case where $(Y_1, \dots, Y_{m+1})$ is a flat
distribution\footnote{For example, the distribution $(Y_{\leq m}, Y_{m+1})
= (f(U), U)$ for a function $f$ and uniform input $U$ is always a flat
distribution even if $f$ itself is not regular.}, then no distribution with the
same support can have higher entropy and in this case
Definitions~\ref{def:block-hard} and \ref{def:ie} coincide as stated in the
following theorem.

\begin{theorem}
	  \label{thm:nb-inacc}
	  Let $(Y_1,, \dots, Y_{m+1})$ be a flat distribution and
	  $\tG$ be an $(m+1)$-block generator consistent with $Y_{\leq m+1}$. Then
	  for $\tY_{\leq m+1} = \tG(\tR_{\leq m+1})$ for uniform $\tR_{\leq m+1}$:
	  \begin{enumerate}
	\item For every $y_{\leq m+1},r_{\leq m+1}\in \Supp(\tY_{\leq m+1}, \tR_{\leq m+1})$, it holds that
	    \begin{align*}
		&\sum_{i=1}^{m+1}\left(\sent{y_i, y_{<i}}{Y_i|Y_{<i}} - \sent{y_i, r_{<i}}{\tY_i|\tR_{<i}}\right)\\
		&\hspace{6em}=\sum_{i=1}^{m+1} \skl{y_i, y_{<i}, r_{<i}}{\tY_i|\tR_{<i}, \tY_{<i}}{Y_i|R_{<i}, Y_{<i}}
	    \end{align*}
		In particular, $(Y_1, \dots, Y_{m+1})$ has next-block inaccessible
		$\klfail$-min relative entropy $(t,\klgap)$ if and only if it has
		inaccessible $\klfail$-max entropy $(t,\klgap)$.
	\item Furthermore,
	    \[
		\sum_{i=1}^{m+1} \left(\ent{Y_i|Y_{<i}} - \ent{\tY_i|\tR_{<i}}\right)
		= \sum_{i=1}^{m+1} \kl{\tY_i|\tR_{<i}, \tY_{<i}}{Y_i|R_{<i}, Y_{<i}},
	    \]
		so in particular, $(Y_1, \dots, Y_{m+1})$ has next-block
		inaccessible relative entropy $(t,\klgap)$ if and only if it has inaccessible
		entropy $(t, \klgap)$.
	  \end{enumerate}
\end{theorem}
\begin{proof}
	For the sample notions, the chain rule (Proposition~\ref{prop:kl-chain})
	gives:
	\begin{displaymath}
		\sum_{i=1}^{m+1}\sent{y_i, y_{<i}}{Y_i|Y_{<i}}
		= \sent{y}{Y_{\leq m+1}} = \log|\Supp(Y_{\leq m+1})|
	\end{displaymath}
	for all $y$ since $Y$ is flat.
	Hence:
	\begin{align*}
		\log\lvert\Supp(Y_{\leq m+1})\rvert - \sum_{i=1}^{m+1} \sent{y_i, y_{<i}}{\tY_i|\tR_{<i}}
		&=\sum_{i=1}^{m+1} \left( \sent{y_i, y_{<i}}{Y_i|Y_{<i}} - \sent{y_i, r_{<i}}{\tY_i|\tR_{<i}}\right)\\
		&= \sum_{i=1}^{m+1} \skl{y_i, y_{<i}, r_{<i}}{\tY_i|\tR_{<i}, \tY_{<i}}{Y_i|R_{<i}, Y_{<i}}
		\;,
	\end{align*}
	so the second claim follows by taking the expectation over $(\tY_{\leq m+1}, \tR_{\leq m+1})$ on both sides.
\end{proof}

By chaining the reductions between the different notions of hardness considered
in this work (hardness in relative entropy, next-block inaccessible relative
entropy and inaccessible entropy), we obtain a more modular proof of the
theorem of Haitner~\etal~\cite{HRVW16}, obtaining inaccessible entropy from any
one-way function.

\begin{theorem}
	Let $n$ be a security parameter, $f:\zo^n\to\zo^n$ be a $(t, \eps)$-one-way
	function, and $X$ be uniform over $\zo^n$. For $\ell\in\{1,\dots,n\}$,
	decompose $f(X) \eqdef (Y_1,\dots, Y_{n/\ell})$ into blocks of length $\ell$. Then:
	\begin{enumerate}
		\item For every $0 \leq \klgap\leq\log(1/\eps)$, $(Y_1, \dots,
			Y_{n/\ell}, X)$ has inaccessible entropy
			$\left(t', \log(1/\eps)-\klgap\right)$ for $t'
			= \Omega\left(t\cdot\klgap\cdot\ell^2/(n^2\cdot 2^\ell)\right)$.
		\item For every $0 < \klfail \leq 1$ and $0 \leq \klgap\leq \log(1/\eps) - \log(2/\klfail)$,
			$(Y_1, \dots, Y_{n/\ell}, X)$ has inaccessible $\klfail$-max entropy
			$\left(t',\log(1/\eps)-\log(2/\klfail)-\klgap\right)$ for
		$t' =\allowbreak\Omega\left(t\cdot\klfail\cdot\klgap\cdot \ell^2/(n^2\cdot 2^\ell)\right)$.
	\end{enumerate}
\end{theorem}
\begin{proof}
	Since $f$ is $(t, \eps)$-one-way, the distributional search problem
	$\big(\Pi^f, f(X)\big)$ where $\Pi^f = \{(f(x), x):x\in\zo^n\}$ is
	$(t, \eps)$-hard. Clearly, $(f(X), X)$ is supported on $\Pi^f$, so by
	applying Theorem~\ref{thm:witness-kl-hard}, we have that $(\Pi^f, f(X), X)$
	has witness hardness $(\Omega(t), \log(1/\eps))$ in relative entropy and $(\Omega(t),
	\log(1/\eps) - \log(2/\klfail))$ in $\klfail/2$-min
	relative entropy.
	Thus, by Theorem~\ref{thm:bkl-hard} we have that $(Y_1, \dotsc, Y_{n/\ell},
	X)$ has next-block inaccessible relative entropy $\left(\Omega\left(t\cdot
			\klgap\cdot \ell^2/(n^2\cdot 2^\ell)\right), \log(1/\eps)
			- \klgap\right)$ and next-block inaccessible $\klfail$-min relative
			entropy $\left(\Omega\left(t\cdot \klfail\cdot \klgap\cdot
			\ell^2/(n^2\cdot 2^\ell)\right), \log(1/\eps) - \log(2/\klfail)
		- \klgap\right)$, and we conclude by Theorem~\ref{thm:nb-inacc}.
\end{proof}

\begin{remark}
	For comparison, the original proof of \cite{HRVW16} shows that for every $0
	< \klfail \leq 1$, $(Y_1, \dots, Y_{n/\ell}, X)$ has inaccessible
	$\klfail$-max entropy $\left(t',\log(1/\eps) - 2\log(1/\klfail)
		- O(1)\right)$ for $t' = \tilde\Omega \left(t\cdot \klfail\cdot
	\ell^2/(n^2\cdot 2^\ell)\right)$, which in particular for fixed $t'$ has
	quadratically worse dependence on $\klfail$ in terms of the achieved
	inaccessible entropy: $\log(1/\eps) - 2\cdot \log(1/\klfail) - O(1)$ rather
	than our $\log(1/\eps) - 1\cdot \log(1/\klfail) - O(1)$.
\end{remark}

\begin{corollary}[Theorem~4.2 in~\cite{HRVW16}]
	Let $n$ be a security parameter, $f:\zo^n\allowbreak \to\zo^n$ be a strong
	one-way function, and $X$ be uniform over $\zo^n$. Then for every $\ell
	= O(\log n)$, $(f(X)_{1\dotsc \ell},\dots,f(X)_{n-\ell + 1 \dotsc n}, X)$
	has inaccessible entropy $\left(n^{\omega(1)}, \omega(\log n)\right)$ and
	inaccessible $1/n^{\omega(1)}$-max entropy $\left(n^{\omega(1)},
	\omega(\log n)\right)$.
\end{corollary}

\section*{Acknowledgements}

We thank Muthuramakrishnan Venkitasubramaniam for an inspiring conversation
which sparked this work.

\bibliographystyle{alpha}
\bibliography{refs,pseudorandomness,crypto}

\appendix
\section{Missing Proofs}\label{app:proof}
\begin{lemma}\label{claim:convex}
	For all $t \geq 1$, $f:x\mapsto\frac{x}{1-(1-x)^t}$ is convex over $[0,1]$.
\end{lemma}

\begin{proof}
	We instead show convexity of $\tilde{f}:x\mapsto f(1-x)$. A straightforward
	computation gives:
	\begin{displaymath}
		\tilde{f}''(x) = \frac{x^{t-2}t\big(t(1-x)(x^t+1)
		- (1+x)(1-x^t)\big)}{(1-x^t)^3}
	\end{displaymath}
	so that it suffices to show the non-negativity of $g(x) = t(1-x)(x^t+1)
	- (1+x)(1-x^t)$ over $[0,1]$. The function $g$ has second derivative
	$t(1-x)(t^2-1)x^{t-2}$, which is non-negative when $x\in[0, 1]$, and thus
	the first derivative $g'$ is non-decreasing. Also, the first derivative at
	$1$ is equal to zero, so that $g'$ is non-positive over $[0, 1]$ and hence
	$g$ is non-increasing over this interval. Since $g(1) = 0$, this implies
	that $g$ is non-negative over $[0,1]$ and $f$ is convex as desired.
\end{proof}

\begin{theorem}[Theorem~\ref{thm:witness-kl-hard} restated]
	Let $\Pi$ be a binary relation  and let $(Y, W)$ be pair of random
	variables supported on $\Pi$.  If $(\Pi, Y)$ is $(t, \eps)$-hard, then
	$(\Pi, Y, W)$ is $(t', \Delta')$ witness hard in relative entropy and $(t',
	\klgap'')$ witness hard in $\klfail$-min relative entropy for every
	$\klfail\in[0,1]$ where $t' = \Omega(t)$, $\klgap' = \log(1/\eps)$ and
	$\klgap'' = \log(\klfail/\invp)$.
\end{theorem}

\begin{proof}
	We proceed similarly to the proof of Theorem~\ref{thm:kl-hard}.
	Let $(\tG, \sS)$ be a pair of algorithms with $\tG = (\tG_1, \tG_\w)$ a two-block generator supported on $\Pi$.
	Define the linear-time oracle algorithm $\sA^{\tG_\w, \sS}(y) \eqdef \tG_\w(\sS(y))$.
	Then
	\begin{align*}
		\pr{\Pi\big(Y, \sA^{\tG_\w, \sS}(Y)\big)}
		&= \pr{\Pi(Y, \tG_\w(\sS(Y)))}\\
		&\geq \pr{\tG_1(\sS(Y)) = Y}\tag{$\tG$ is supported on $\Pi$}\\
		&= \sum_{r\in\Supp(\tR)}
		\pr{\sS(Y)=r\wedge Y=\tG_1(r)}\\
		&\geq \sum_{\substack{r\in\Supp(\tR)\\w\in\Supp(\tG_2(\tR))}}
		\pr{\sS(Y)=r\wedge Y=\tG_1(r)\wedge W=w}\\
		&= \ex[\substack{r\from\tR\\ w\from\tG_2(r)}]{\frac{\pr{\sS(Y)=r \wedge
		Y = \tG_1(r) \wedge W = w}}{\pr{\tR=r \wedge \tG_2(r) = w}}}\\
		&= \ex[\substack{r\from\tR\\ y\from \tG_1(r)\\ w\from\tG_2(r)}]
		{2^{-\skl{r, y, w}{\tR, \tG_1(\tR), \tG_2(\tR)}{\sS(Y), Y, W}}},
	\end{align*}
	The witness hardness in relative entropy then follows by applying Jensen's
	inequality (since $2^{-x}$ is convex) and the witness hardness in
	$\klfail$-min relative entropy follows by Markov's inequality by
	considering the event that the sample relative entropy is smaller than
	$\klgap$ (this event has density at least $\klfail$).
\end{proof}

\end{document}